\documentclass{article}
\usepackage{microtype}
\usepackage{graphicx}
\usepackage{subfigure}
\usepackage{booktabs} 
\usepackage{hyperref}
\usepackage{amsmath, amsthm}
\usepackage{booktabs}
\usepackage{color}
\usepackage{todonotes}
\usepackage{amsfonts}
\usepackage[margin=0.75in]{geometry}
\usepackage{rotating}
\newtheorem{theorem}[]{Theorem}

\newtheorem{lemma}[]{Lemma}

\DeclareMathOperator{\ima}{Im}
\DeclareMathOperator{\rank}{rank}
\begin{document}

%\twocolumn[
%\icmltitle{A Persistent Homology Approach to Time Series Classification}

%\icmlsetsymbol{equal}{*}

%\begin{icmlauthorlist}
%~
%\end{icmlauthorlist}

\title{A Persistent Homology Approach to Time Series Classification}
\date{}
\author{Yu-Min Chung$^*$, William Cruse\footnote{University of North Carolina at Greensboro, Department of Mathematics and Statistics}, and Austin Lawson\footnote{University of North Carolina at Greensboro, Program of Informatics and Analytics}}
\maketitle

%\icmlkeywords{Time Series, Topological Data Analysis, Persistent Homology}

%\vskip 0.3in
%]

%\printAffiliationsAndNotice{\icmlEqualContribution}.

\begin{abstract}

Topological Data Analysis (TDA) is a rising field of computational topology in which the topological structure of a data set can be observed by persistent homology. 
By considering a sequence of sublevel sets, one obtains a filtration that tracks changes in topological information.  These changes can be recorded in multi-sets known as {\it persistence diagrams}. 
Converting information stored in persistence diagrams into a form compatible with modern machine learning algorithms is a major vein of research in TDA.
{\it Persistence curves}, a recently developed framework, provides a canonical and flexible way to encode the information presented in persistence diagrams into vectors.
In this work, we propose a new set of metrics based on persistence curves. We prove the stability of the proposed metrics.  Finally, we apply these metrics to the UCR Time Series Classification Archive. These empirical results show that our metrics perform better than the relevant benchmark in most cases and warrant further study.

\end{abstract}

\section{Introduction}
In recent years, Topological Data Analysis (TDA) has expanded rapidly. With applications in  health research \cite{li2015identification}, 3D modeling \cite{pht}, and aviation \cite{li2019topological}, the field is poised to become a staple of data analytics. Among the tools available to a topological data analyst is persistent homology. This tool considers a dataset at all scales simultaneously via sequence of topological spaces and tracks when topological features appear (are born) and disappear (die) within the sequence. These birth-death times can be collected and stored in a multi-set known as a persistence diagram; however, these diagrams are not compatible with modern machine and deep learning algorithms. This discrepancy has lead to several projects dedicated to summarizing persistence diagrams in such a way that preserves the topological information and is compatible with machine and deep learning. Summaries such as the persistence images \cite{adams2017persistence}, the persistence landscapes \cite{bubenik2015statistical}, and the persistent entropy \cite{persistentEntropy}. The particular interest in this paper is the persistence curve framework \cite{chunglawson2019}. In that work, they showed this framework capable of generating new summaries as well as several well-known summaries such as persistence landscapes and persistent entropy.  Recently, there has been some work on applying TDA tools to classify time series.  For instance, in these work  \cite{umeda2017time}, \cite{venkataraman2016persistent}, and \cite{seversky2016time}, authors used lag map method to reconstruct the attractors and utilized TDA tools (more precisely, the Rips complex filtration) to extract features from those attractors. In \cite{marchese2018signal}, authors defined a new distance on the space of persistence diagram and used this new distance to classify time series.  In \cite{wang2018topological} and \cite{chung2019persistent}, authors used TDA tools to analyze physiological signals. 

In this paper, our main contribution is to use persistence curves to define two ensemble metrics for the difference of two time series. One of the ensemble metrics uses the Euclidean 1-norm while the other uses the Dynamic Time Warping (DTW) Distance between time series introduced in \cite{berndt1994using}. In Section \ref{sec:background}, we provide a gentle introduction to persistent homology and persistence curves and explain how we extract these objects from time series. In Section \ref{sec:PCDTW}, we introduce the topological transformation of a time series, and the definition of our proposed metrics is presented in Section~\ref{sec:new metric}.
Finally, in Section \ref{sec:applications}, we apply these new distances to the UCR Time Series archive \cite{UCRArchive2018} via a 1-NN test and compare our results to the reported benchmark in the database.

\section{Mathematical Background}\label{sec:background}
We will lay out the basic mathematical concepts to equip the reader with the tools to understand the models developed later. We begin by discussing time series and Dynamic Time Warping before moving on to homology, persistent homology, and persistence curves. 
\subsection{Time Series and Dynamic Time Warping}
For a natural number $n\in\mathbb{N}$ let $[n] = \{0,1,\ldots,n\}$. A \textbf{time series} $s$ is a function $s:[n]\to\mathbb{R}$. We commonly write $s_i$ for $s(i)$. Given a time series $s$ we can produce a continuous piecewise linear functional representation $f_s$ of $s$ by linearly interpolating between the points of $s$. One can show in this setting that the 1-norm of a time series bounds the 1-norm of its functional representation from above and similarly for the $\infty$-norm. That is $\|s\|_1\ge \|f_s\|_1$ and $\|s\|_{\infty}\ge \|f_s\|_\infty$ where $\|s\|_1 = \sum_{i=1}^n|s_i|, \|s\|_\infty = \max(|s|), \|f_s\|_1 = \int_{0,n}|f_s(x)\mathrm{d}x, \|f_s\|_\infty = \max(|f_s|)$. Moreover, we have $\|s\|_{\infty}\le \|s\|_1$ . Note these inequalities are achievable due to the fact that $f_s$ is continuous, bounded, and defined on a compact domain. { To prove that the discrete 1-norm is larger than the continuous 1-norm suppose we have a non-negative time series $s = (s_i,s_{i+1})$. The corresponding piecewise linear function is $f_i(x) = (s_{i+1} - s_i)(x-i)+s_i$. Integrating this function from $i$ to $i+1$ gives $A_i = (s_{i}+s_{i+1})/2$; however, $\|s\|_1 = s_{i}+s_{i+1}$. Consider $t = (t_0,\ldots t_{i+1})$. The area under the corresponding piecewise linear curve $f$ is given by $\|f\|_1 = A_0 + A_1 + \ldots + A_{i}$. But this is exactly $\|f\|_1 = t_0/2 + t_1 + \ldots + t_i + t_{i+1}/2 \le t_0 + t_1 + \ldots + t_i + t_{i+1} = \|t\|_1$. So if $s$ is any time series with corresponding function $f$ then one may consider the time series $t$ where $t_i = |s_i|$. This is a nonnegative time series whose corresponding function is exactly defined by $g(x) = |f(x)|$. Moreover, $\|s\|_1 = \|t\|_1$ and $\|f\|_1 = \|g\|_1$. The previous argument applies and we see $\|s\|_1\ge\|f\|_1$}

We will now discuss Dynamic Time Warping \cite{berndt1994using}. For the rest of this section, let $s$ and $t$ be time series defined on $[n]$ and $[m]$ respectively. 
A \textbf{warping path} $w$ between $s$ and $t$ is a function $w:[N]\to[n]\times[m]$ so that $w(1) = (0,0), w(N) = (n,m)$ and $w_{i+1}-w_i\in\{(1,0),(1,1),(0,1)\}$,  where $N$ is the length of the path from $(0,0)$ to $(n,m)$. We denote the collection of warping path between $s$ and $t$ by $\Omega_{s,t}$.  A warping path induces two time series $s'$ and $t'$ both defined on $[N]$ where $s'_i = s_{\pi_1(w(i))}$ and $t'_i = t_{\pi_2(w(i))}$, where $\pi_i$ is the projection to the $i$-th coordinate for $i=1,~2$. 

Let $k:\mathbb{R}^2\to\mathbb{R}$ be non-negative. The cost of a warping path $w$ between $s$ and $t$ can be defined as $K(w) = \sum_{i=1}^N k(s'_i,t'_i)$. In this paper, we choose $k(x,y) = |x-y|$. For instance, suppose $n=m$, and consider the ``diagonal'' path $w_0 = \{ (i,i) \}_{i=1}^n$. Then 
$K(w_0) = \sum_{i=1}^{n} k(s_i, t_i) = \sum_{i=1}^{n} | s_i - t_i| = \| s - t\|_1.$

The \textbf{Dynamic Time Warping distance} (DTW distance) between $s$ and $t$ is defined to be the minimum cost over all warping paths. That is, 
\begin{equation}
    \|s-t\|_{DTW} := \min_{w\in \Omega_{s,t}} K(w).
    \label{def:dtw distance}
\end{equation}
There is a relation between DTW distance and Euclidean distance. Observe that since $w_0 \in \Omega_{s,t}$, one has 
\begin{equation}
    \label{equ:relation between dtw and 1norm}
    \|s-t\|_{DTW} \leq \| s - t\|_1.
    \end{equation}

It is well known that the functional representations of any two time series - two 1-dimensional lines without loops - are homeomorphic, i.e. either one may be continuously morphed into the other without punctures or tears. Heuristically, one can view the DTW distance measure of cost for that morphing. 

\subsection{Persistent Homology}
Homology is a classic subject in mathematics that provides information about a topological space that is invariant under continuous deformations. This means that the homology of a space gives us useful information about its topological structure.
Informally, the $k$-th homology group $H_k(X)$ of a space $X$ is often used to count the \textbf{{$k$-th Betti}} numbers, i.e. the number of $k$-dimensional holes of $X$. For example,  $H_0(X)$ counts the connected components and $H_1(X)$ counts loops, $H_2(X)$ counts voids or ``air pockets'', and so on.

 Persistent homology, first defined in \cite{edelsbrunner2000topological} is a powerful tool from the field of topological data analysis that tracks the changes of homology over a filtration.

A \textbf{filtration} of a space $X$ is an increasing sequence of spaces $\emptyset = X_0\subset X_1\subset\ldots\subset X_n = X$.  One may obtain homology groups for each $X_i$.  Because of the subset relations, we can track the changes of homology groups. This process is what we call persistent homology. For a time series $s$ with corresponding continuous piecewise linear function $f_s$, we can obtain a filtration by considering sublevel sets $f_s^{-1}((-\infty, x]) = \{y\in \mathbb{R})\mid f_s(y)\in(-\infty,x]\}$. The inclusion $X_t\subset X_s$ induces a map $g^k_{t,s}:H_k(X_t)\to H_k(X_s)$ between the homology groups. We note that homology in this sequence only changes at critical values of $f$. \cite{edelsbrunner2010computational}

 We say a homology class $\alpha$ is \textbf{born} at $b$ if we have $\alpha\in H_k(\mathcal{X}_{b})$ and $\alpha\notin \ima g^k_{b-1,b}$.  We say that $\alpha$ born at $b$ \textbf{dies} at $d$, $d\ge b$ if $g^k_{b,d-1}(\alpha)\notin \ima g^k_{b-1,d-1}$, but $g^k_{b,d}(\alpha)\in \ima g^k_{b-1,d}$, i.e. if it merges with a previous class. The ranks $\beta^k_{b,d} = \rank\ima g^k_{b,d}$ for $d\ge b$ form the \textbf{persistent Betti numbers} of the filtration. 
 
 These persistent Betti numbers count the number of classes that were born at or before $b$ and are still alive at $d$. Inclusion-exclusion allows us to count exactly the number $\mu^k_{b,d}$ of classes born at $b$ and die at $d$ by $\mu^k_{b,d} = \beta^k_{b,d-1} - \beta^k_{b-1,d-1} + \beta^k_{b-1,d} - \beta^k_{b,d}$. The \textbf{$k$-th persistence diagram}, or just \textbf{diagram}, $\mathcal{P}_k(f)$ associated to the filtering function $f$ of a space $X$ is a multi-set, that is a set of points with multiplicity, of birth-death pairs $(b,d)$ with multiplicity $\mu^k_{b,d}$ along with the diagonal points $(b,b)$ each with infinite multiplicity. To shorten notation, we will often represent persistence diagrams with the letter $D$.
 
The stability theorem for persistence diagrams, which first appeared in \cite{cohen2007stability}, states that if $f_1, ~f_2$ are functions that have finitely many critical values then \[W_\infty(\mathcal{P}_k(f_1),\mathcal{P}_k(f_2))\le \|f_1-f_2\|_\infty,\]
where $W_\infty(\mathcal{P}_k(f_1),\mathcal{P}_k(f_2))$ is known as the \textbf{Wasserstein $\infty$-metric}, or \textbf{bottleneck distance} defined as

\[W_\infty(D_1,D_2) = \inf\limits_{\stackrel{\hbox{bijections }}{{\eta:D_1\to D_2}}}\sup_{x\in D_1}\|x-\eta(x)\|,\]
where $D_1 = \mathcal{P}_k(f_1)$, $D_2= \mathcal{P}_k(f_2)$.

 This stability theorem tells us that if the original space is altered by a small amount, then the diagrams are altered by a small amount. Because diagrams are not readily compatible with machin learning, we need to summarize them in some way. In the next section, we will discuss the persistence curve framework, which is capable of generating summaries of diagrams.

\subsection{Persistence Curves}
Persistence Curves, first defined in \cite{chunglawson2019}, draws inspiration from the Fundamental Lemma of Persistent Homology, which states that the Betti number of an element in the filtration $\beta_k(f^{-1}((-\infty,x])$ is given by $\sum_{i\le x}\sum_{j>x}\mu^k_{a_i,a_j}$. The statement here says that the $k$-th Betti number of the subspace $f^{-1}((-\infty,x])$ is given exactly by the number of points in the box formed by the points up and to the left of the diagonal at $x$.

Persistence curves take this concept and generalize it. Instead of just counting the points, we will place a function over the points and compute some statistic over the aforementioned box. The idea is formalized follows.

Let $D$ be a diagram and let $x,y,z\in\mathbb{R}$. Suppose $\psi(D,x,y,z)$ is a real-valued function so that $\psi(D,x,x,z) = 0$. Let $T$ be a statistic, or a function of multi-set. We can define a persistence curve as:
\[P(D,\psi,T)(x) = T(\{\psi(D,b,d,x)\mid b\le x < d \}.\]
We refer to the set $\{(b,d)\mid b\le x < d\}$ as a \textbf{fundamental box} at $x$. Figure~\ref{FLPH} illustrates a fundamental box.
As a quick example, if $T=\Sigma$ is the sum statistic and if $\mathbf{1}(D,x,y,z) = 1$ when $x\neq y$ and $0$ otherwise, then $P(D,\mathbf{1}, \Sigma)$ produces the Betti number curve. That is to say $P(D,\mathbf{1}, \Sigma)(x) = \beta(x)$. 

Generally, one chooses $\psi$ to be a function that carries useful information about the diagram points.  In this paper, we are interested in one particular curve called the stabilized life curve. If $\ell(x,y) = y-x$ then with $\psi(D,x,y,t) = \frac{\ell(x,y)}{\sum_{(b,d)\in D)}\ell(b,d)}$ we define the stabilized life curve as $\mathbf{sl} \equiv P(D,\psi,\Sigma)$. Figure~\ref{fig:PCWorkflow} illustrates the workflow of beginning with a time series, then computing the diagram via sublevel sets before finally computing the stabilized life curve. The top three images constitute a toy example in which the birth and death coordinates are connected to the cooresponding time series points.
\begin{figure}
    \centering
    \includegraphics[width=\textwidth]{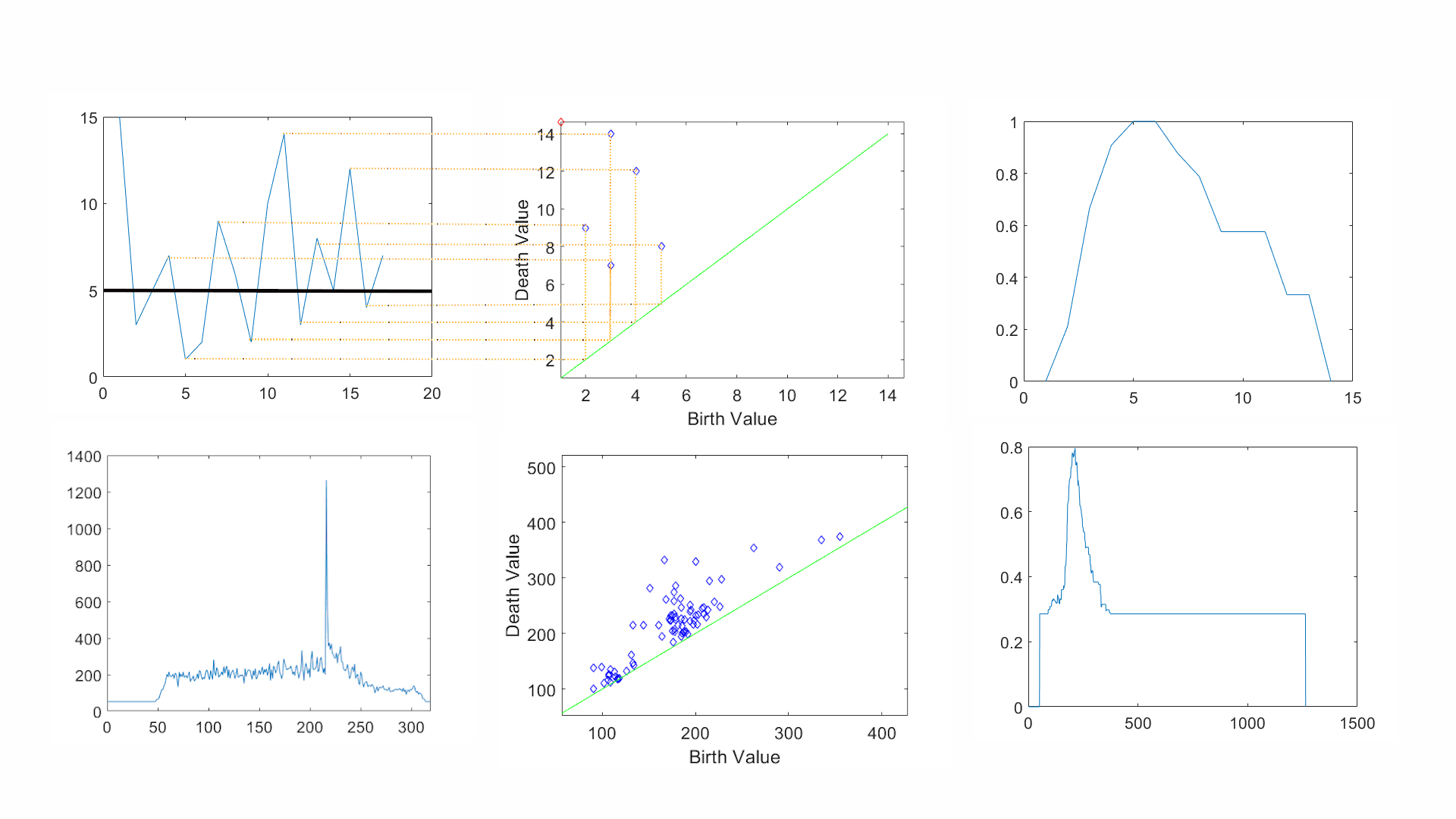}
    \caption{Two examples of the persistence curve workflow.} 
    \label{fig:PCWorkflow}
\end{figure}

A general stability bound for persistence curves was established in \cite{chunglawson2019}, which we state below.
\begin{theorem}\label{pcbound}\cite{chunglawson2019}
Let $D_1,D_2$ be persistence diagrams. Let $T=\Sigma$ be the sum statistic. Suppose that $T(\emptyset) = 0$. Let $\psi(D_1, \cdot)$ and $\psi(D_2, \cdot)$ be continuous functions. Then there exists a constant $C$ such that the following estimate holds
\begin{align}
    \| P(D_1, \psi, \Sigma) - P(D_2, \psi, \Sigma)\|_1 \leq &CW_\infty(D_1,D_2).
\end{align}
\end{theorem}
In essence, the constant $C$ depends on the function $\psi$, and number of points in $D_1$ and $D_2$.  \cite{chunglawson2019} provides an explicit form of $C$. Since we will not use the its explicit form in this work, we refer interested readers to \cite{chunglawson2019} for more details about the constant $C$.

In practice, we need to discretize the persistence curves. Given a persistence curve $P$ and a mesh of $M$ points $\{x_i\}_{i=1}^{M}\subset \mathbb{R}$. The corresponding sequence $\bar{P} = \{P(x_i)\}$ is called a \textbf{discretized persistence curve}. As the number of points in the mesh increases, the discretized persistence curve approaches its continuous counterpart. Moreover, the proof of the stability for the continuous version can be altered to provide a proof for the same stability bound for the discretized curves, provided that the mesh contains all birth and death values of the diagram.

\section{Topological Transformation of Time Series}\label{sec:PCDTW}
Our main method transforms a time series into a set of persistence curves.  This set of curves will be a topological representation of time series. The focus of this section is to show that such topological transformation is stable with a given metric.  More precisely, a small perturbation in the original time series leads to a small change in the persistence curve.  We prove stability results with two different metrics.

Let $s$ be a time series.  We consider the usual continuous, piece-wise linear function representation $f_s$ of $s$. We also consider a persistence diagrams of $f_s$, $D(f_s)$, by sublevel set filtration. Lastly, we use a persistence curve to summarize $D(f_s)$. To ensure robustness, we provide a stability bound for the DTW distance on the discretized persistence curve. Prior to this, we show an interesting result bounding the bottleneck distance by the DTW distance.
\begin{lemma}\label{lem:bottleneckDTW}
Let $s,t$ be time series and suppose $f_s,f_t$ be their corresponding continuous piece-wise linear functions. Let $\{x_i\}_{i=1}^n$ be a mesh fine enough to contain all birth and death values of $D(f_s)$ and $D(f_t)$. Then $$W_\infty(D(f_s),D(f_t))\le \|s-t\|_{DTW}.$$
\end{lemma}
\begin{proof}
let $w^*$ be the optimal path satisfying DTW distance between $s$ and $t$. Denote the induced time series by $s^* = \{ s(\pi_1(w^*(i)) \}_{i=1}^{N}$ and $t^* = \{ t(\pi_2(w^*(i)) \}_{i=1}^{N}$, where $N$ is the length of the optimal path $w^*$,

and corresponding functions $f_{s^*}$ and $f_{t^*}$. By the definition of a warping path, namely that $w(i+1) - w(i) \in\{(1,0),(1,1),(0,1)\}$ we have that $w$ can change in either coordinate by at most one and it cannot decrease. {The proof of this is extraordinarily tedious. How about some heuristic like: The critical values of $f_s$ occur exactly at values of the time series. Let $s_i$ be the height of a critical value that causes a change in the homology of the sublevel set filtration. Let $j>i$ be the smallest index for which $s_i\neq s_j$ and let $k<i$ be the largest index so that $s_i\neq s_k$. Then either  $s_i>s_k$ and $s_i>s_j$ or $s_i<s_k$ and $s_i<s_j$.  By definition of DTW, there is some $i^*$ for which $s_i = s^*_{i^*}$. Moreover, because $w^*$ can only change the first coordinate by at most one, the closest indices below and above $i^*$ for which heights are difference from $s^*_{i^*}$ have heights $s_j$ and $s_k$. That is to say that $s_{i^*}^*$ is the height of a critical value that changes homology in the sublevel set filtration. A similar argument shows that these critical values for $f_{s^*}$ are critical values of $f_s$. } 
Therefore, $f_{s^*}$ has the same critical points as $f_s$ hence $D(f_s) = D(f_{s^*})$.  The same argument shows $D(f_t) = D(f_{t^*})$.  From the definition of DTW distance, one has $\|s^*-t^*\|_1 = \|s-t\|_{DTW}$. Once we put it all together we have the following string of inequalities:
\begin{align*}
    W_\infty(D(f_{s^*}),D(f_{t^*})) &\le \|f_{s^*}-f_{t^*}\|_{\infty}\\
    &\le \|s^*-t^*\|_{\infty}\\
    &\le \|s^*-t^*\|_1 = \|s-t\|_{DTW}.
\end{align*}
\end{proof}
This lemma, to the best of our knowledge, is the first result that connects persistence diagram of sublevel set filtration of 1D function and DTW.  It would be interesting to investigate relation between optimal match in Bottleneck distance and warping path in DTW.

We are now ready to present our main result.

\begin{theorem}
\label{thm:main}
Let $\psi$ and $T=\Sigma$ be the sum statistic. Let $s,~t$ be time series and suppose $f_s,~f_t$ be their corresponding continuous piecewise linear functions. Let $\{x_i\}_{i=1}^n$ be a mesh fine enough to contain all birth and death values of $D(f_s)$ and $D(f_t)$. 
 
Then there is a constant $C$, which is the same constant as in Theorem~\ref{pcbound}, such that% 
\begin{equation}
    \|\bar{P}(D(f_s),\psi,T) - \bar{P}(D(f_t),\psi,T)\|_{X}\le C\|s-t\|_{X},
\end{equation}
where $X = 1$ or $X = DTW$. 
\end{theorem}
\begin{proof}
 When $X = 1$, the proof is the direct result of the general bound persistence curves and the stability result for the bottleneck distance on persistence diagrams \cite{cohen2007stability}. Note that for any time series $s$, $f_s$ satisfies this regularity condition due to it being continuous and having finitely many critical values, hence we benefit from the diagram stability theorem.

To see the case of $X = DTW$,
we first ease notation by writing $D_s$ for $D(f_s)$, $D_t$ for $D(f_t)$, $\bar{P}_s$ for $\bar{P}(D_s,\psi,T)$ and $\bar{P}_t$ for $\bar{P}_t(D_t,\psi,T)$.

Then by \eqref{equ:relation between dtw and 1norm}, Theorem~\ref{pcbound}, and Lemma~\ref{lem:bottleneckDTW}, respectively, we have
\begin{align*}
    \|\bar{P}_1 - \bar{P}_2\|_{DTW}&\le \|\bar{P}_1 - \bar{P}_2\|_1\\
    &\le CW_\infty(D_1,D_2)\\
&    \le C\|s-t\|_{DTW}.
\end{align*}
\end{proof}

We summarize all relations among different metrics that we have discussed in this paper:  
\begin{align*}
    \|P_s - P_t \|_{DTW}\leq \|P_s - P_t\|_1 &\leq C W_{\infty}(D_s, D_t) \\ &\leq C \|s-t\|_{DTW}\\& \leq C \|s-t\|_1.
\end{align*}

\begin{figure}
    \centering
    \subfigure[A successful transformation]{\includegraphics[width=0.75\textwidth]{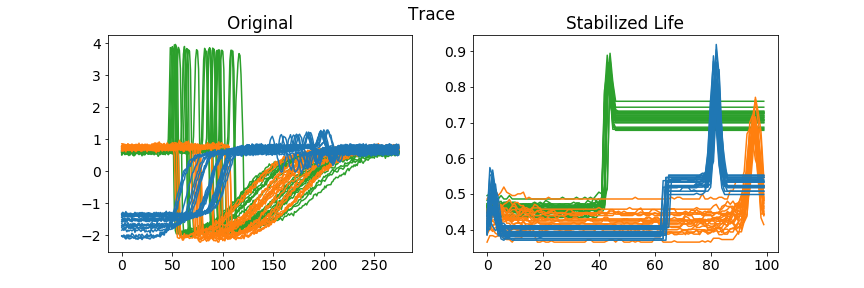}}
    \subfigure[An unsuccessful transformation]{\includegraphics[width=0.75\textwidth]{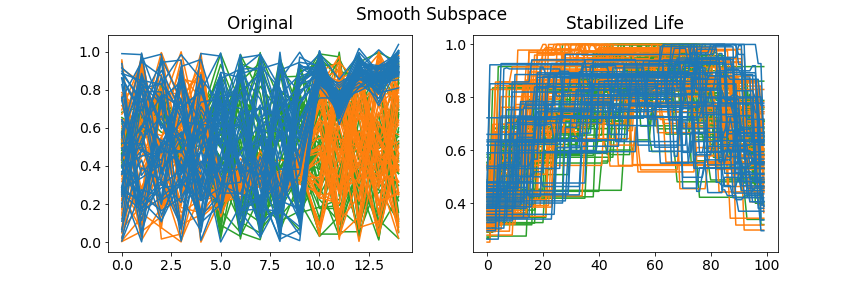}}
    \caption{Two examples of transformation via $\mathbf{sl}$. }
    \label{fig:transformationexample}
\end{figure}

We conclude this section by showing an example of proposed topological transformation. Figure \ref{fig:transformationexample} shows two examples of the result when we use the stabilized  life curve.   
We consider the ``Trace'' and ``SmoothSubspace'' dataset as shown in Figure~\ref{fig:transformationexample}(a) and (b) respectively.  Each dataset consists of three different classes which are shown in the blue, orange, and green, respectively.
Notice on the left of (a) which shows the original time series, two classes of time series (orange and green ones noticeably. Moreover, the members of each class seem to be a horizontal translation of the other members. On the right of (a), we see the result of the topological transformation. Observe that not only does the transformation separate the classes, but also fixes the translation. In (b), we see on the left the original and on the right the transformation. Unfortunately, after the transformation the overlap clearly becomes worse.  These examples motivate us to consider both geometric and topological information, and support exploration of the ensemble metric that we propose in the next section.

\section{New Metrics}
\label{sec:new metric}
Persistence curves encode topological information into a format that is compatible with machine learning algorithms. By utilizing the original space in combination with the topological transformation, we introduce an ensemble metric and define a new distance between time series. Let $\alpha\in(0,1]$. Let $\bar{P}(\cdot,\psi,T)$ be a discretized persistence curve. If $s$ and $t$ are time series with corresponding continuous functions $f_s$ and $f_t$. Finally let $\bar{P}_s = \bar{P}(D(f_s),\psi,T)$ and $\bar{P}_t = \bar{P}(D(f_t),\psi,T)$.

Define the distance between $s$ and $t$ to be \[d^{P}_{1}(\alpha; s,t) = \alpha\|s-t\|_1 + (1-\alpha)\|\bar{P}_s - \bar{P}_t\|_1.\]
Note that if $\alpha = 0$, then this is not a true metric of $s$ and $t$ since we cannot guarantee that $\|P(D(f),\psi,T) - P(D(f_t),\psi,T\|_1 = 0$ implies $s=t$.  We can define a similar metric by using the DTW distance: \[d^{P}_{DTW}(\alpha;s,t) = \alpha\|s-t\|_{DTW} +
(1-\alpha)\|\bar{P}_s - \bar{P}_t\|_{DTW}.\]

These hybrid distances take both topological and geometric information into account.  This new set of metrics contains two parameters.  The first parameter is the $P$, persistence curves. One may choose $P$ based on their prior knowledge about the dataset. \cite{chunglawson2019} provides a list of persistence curves. The other parameter, $\alpha$, is a weighting of the distance between the original time series and the distance between their topological transformations. Based on different datasets, we believe $\alpha$ can be optimized. In this work, we use $\displaystyle d^{\mathbf{sl}}_1$ and $\displaystyle d^{\mathbf{sl}}_{DTW}$ and  $\alpha \in \{0,0.25,0.5,0.75,1\}$ as demonstration. 

\section{Applications to Time Series Classification}
\label{sec:applications}
The UCR Time Series Archive \cite{UCRArchive2018} contains 128 datasets of time series gathered by different means for different scenarios and contributed by several different people. each dataset varies in train/test size and time series length. The task for each of these datasets is classification, and each dataset has a different number of classes. Some of the datasets were z-normalized and others were not. In our analysis, we did not attempt to alter the original time series in any other way than what was necessary to compute the topological transformation. The archive has reported three different benchmark models: Euclidean 2-norm distance 1-NN, DTW distance with a learned warping window 1-NN and DTW distance with a warping window of size 100 plut 1-NN. The administrator of the archive recommends those who are proposing new distance measurements to also perform the 1-NN test to allow for a proper comparison between the metrics.
We performed our analyses in Python 3.7.1. We use the \texttt{Perseus} \cite{perseus} software to compute persistence diagrams. In particular we make use of the \texttt{cubtop} function of \texttt{Perseus}. To use Perseus, one must input positive integers. Thus, to compute persistence, we transformed the time series in the following way: given a time series $s$, we fill missing values with 0 and then map $s$ to a new time series $s'$ via the formula $s'_x = \mathrm{Round}\left(1+\frac{s_x - \min(s)}{\max(s)-\min(s)}\cdot 100\right)$ for $x\in[n]$.  Then by feeding this new time series into \texttt{Perseus}, we obtain a persistence diagram. We also note that due to the transformation $s'$, we obtain a natural mesh $\{0,1,\ldots 100\}$ to compute a discretized persistence curve via the \texttt{PersistenceCurves} python package \cite{pcpackage}. Note that by design, this mesh captures all possible birth and death values of diagrams arising from this process. Once we obtain a diagram $D_{s'}$, we compute the stabilized life, $\mathbf{sl}(D_{s'})$. For each folder in the UCR time series database, we performed 1-NN classifications with the $d^{\mathbf{sl}}_1$ and $d^{\mathbf{sl}}_{DTW}$.
\begin{figure}
    \centering
    \includegraphics[width=0.5\textwidth]{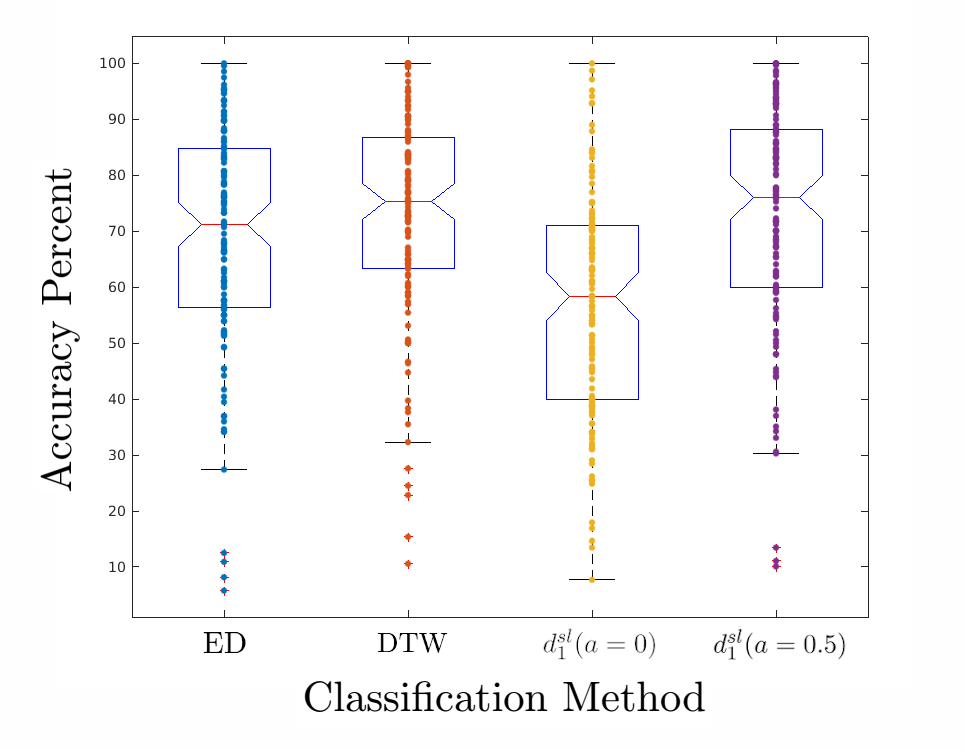}
    \caption{Classification scores for each distance.}
    \label{fig:boxplot}
\end{figure}

Figure \ref{fig:boxplot} shows a boxplot of scores for each model considered. Notice that the as suggested in Section~\ref{sec:PCDTW}, the persistence curve alone $(\alpha=0)$ does not perform well. We can see a comparison of each model and that the ensemble models perform well.

We moved on to testing the usefulness of these two metrics by measuring their performance on the UCR Time Series Archive via the 1-NN test. For the $d^{\mathbf{sl}}_1(\alpha=0.5)$, we tested all 128 datasets in the archive. We found that the metric outperformed the reported Euclidean distance (ED) benchmark on the UCR Time Series by $97$ out of $128$ datasets as shown in Figure~\ref{fig:EnsembleEDScores}(a).  Each plot in Figure \ref{fig:EnsembleEDScores} shows the difference in the accuracy scores between the given models for each of the UCR datasets in alphabetical order.  For more details about performance of individual datasets, we include the Excel file of our main results in the supplementary file.  A similarly good result came from weighting the topological features higher as shown in Figure~\ref{fig:EnsembleEDScores}(c).  

If we compare this metric ($d^{\mathbf{sl}}_1 (\alpha=0.5)$) to the UCR reported DTW (learned\textunderscore{w}) benchmark, our results are better $43$ out of $128$ times and $41$ times for $d^{\mathbf{sl}}_1 (\alpha=0.25)$ as shown in Figure \ref{fig:EnsembleEDScores}(b) and (d). 

The DTW ensemble metric was also tested on all $128$ UCR datasets.
In these experiments, we did not make use of a warping window and instead calculated the DTW distance by using the full time series. Out of these 128 datasets, our results showed $d^{\mathbf{sl}}_{DTW}(\alpha=0.5)$ are better than DTW (learned\textunderscore{w}) 61 times and tied 3 times.

When using $d^{\mathbf{sl}}_{DTW}(\alpha=0.75)$, our results are better than the calculated DTW  64 times, tied 3 times. Figure \ref{fig:EnsembleDTWScores} shows the differences in accuracy between the given models. We see the topological information is competitive with the learned warping window. It would be interesting to combine these two methods.

\begin{figure}
    \centering
    \subfigure[$\alpha=0.5$]{
    \includegraphics[width=0.23\textwidth]{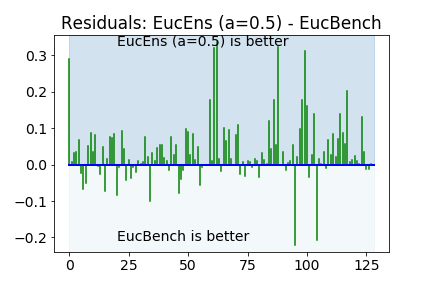}}
    \subfigure[$\alpha=0.5$]{\includegraphics[width=0.23\textwidth]{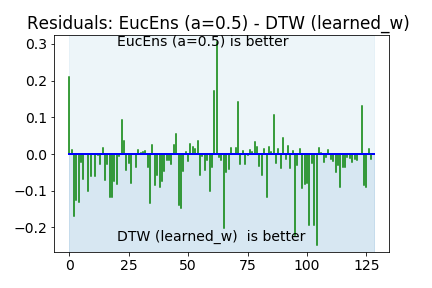}}
    \subfigure[$\alpha = 0.25$]{\includegraphics[width=0.23\textwidth]{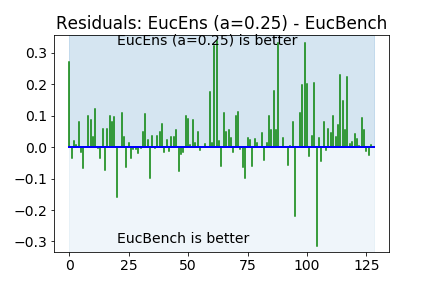}}
    \subfigure[$\alpha = 0.25$]{\includegraphics[width=0.23\textwidth]{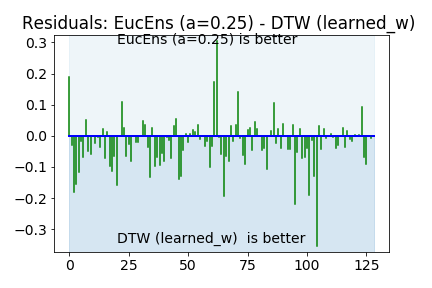}}
    \caption{Residual plots for difference in accuracy between ensembles ($d^{\mathbf{sl}}_1(\alpha)$) and benchmarks}\label{fig:EnsembleEDScores}
\end{figure}

\begin{figure}
    \centering
    \subfigure[$\alpha = 0.5$]{
    \includegraphics[width=0.23\textwidth]{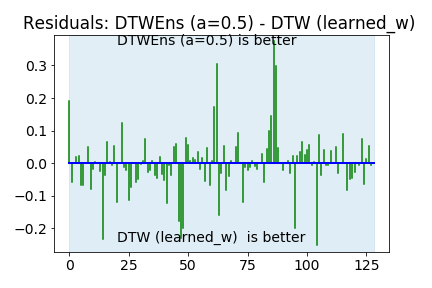}}
    \subfigure[$\alpha = 0.5$]{\includegraphics[width=0.23\textwidth]{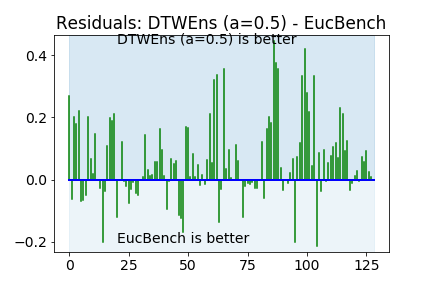}}
    \subfigure[$\alpha = 0.75$]{\includegraphics[width=0.23\textwidth]{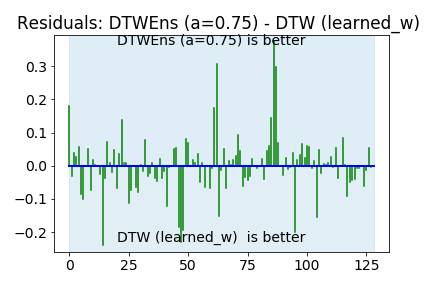}}
    \subfigure[$\alpha = 0.75$]{\includegraphics[width=0.23\textwidth]{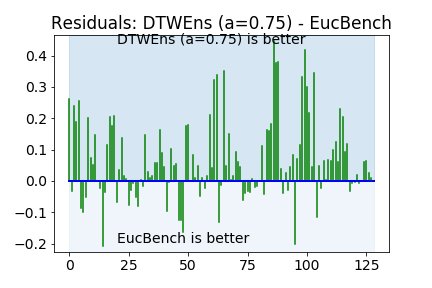}}
    \caption{Performance of DTW ensemble metric ($d^{\mathbf{sl}}_{DTW}(\alpha)$) with $\alpha\in\{0.5,0.75\}$  on 35 UCR datasets}\label{fig:EnsembleDTWScores}
\end{figure}

\section{Conclusion}
With the growing popularity of topological data analysis and computational topology as a whole, more researchers seek to understand data by studying it's shape. In this work, we proposed a transformation of time series based on topology and in particular persistence curves. From here we proposed two families ensemble metrics that measures distance between two time series by considering the time series itself as well as their topological transforms and computes a weighted average of either the Euclidean or DTW distance on the respective parts. We provided a general stability result for the DTW distance as well as a concrete result for the stabilized life curve, which was the persistence curve we used in our experiments. 

These initial results are promising and pave the way for future exploration through the addition of not only more datasets from the archive, but also for the inclusion of a warping window for DTW distance. Moreover, we focused only on the stabilized life curve in this work. The general nature of the persistence curve framework allows for a rich toolbox of curves to define or select from. With this in mind, we can ask the deeper question of curve selection. That is, given a dataset, is it possible to systematically decide, perhaps with a given family of persistence curves, which curve will perform optimally relative that family. On a broader level, we can ask to what extent can we develop a statistical/machine learning theory around this framework. These questions point to a multitude of research and application avenues in the future.

\bibliography{example_paper}
\bibliographystyle{plain}

\end{document}